\documentclass{article}
\usepackage{setspace, amsmath, amssymb, amsthm}
\usepackage[colorlinks = true]{hyperref}
\usepackage[normalem]{ulem}
\usepackage[margin = 2.5cm]{geometry}
\usepackage[english]{babel}
\usepackage[utf8x]{inputenc}
\usepackage{amsmath}
\usepackage{graphicx}
\usepackage[usenames,dvipsnames]{pstricks}
\usepackage{epsfig}
\usepackage{pst-grad} 
\usepackage{pst-plot} 
\usepackage{epstopdf}
\usepackage{caption}
\usepackage{subcaption}
\usepackage{hyperref}
\usepackage[ruled,linesnumbered]{algorithm2e}
\newtheorem{theorem}{Theorem}[section]

\newtheorem{corollary}[theorem]{Corollary}
\newtheorem{definition}{Definition}
\newtheorem{mydef-propostion}{Propostion}
\newtheorem{example}{Example}
\title{Algorithm for $\mathcal{B}$-partitions, parameterized complexity of the matrix determinant and permanent \footnote{The preliminary idea of the work was presented in `18th Haifa Workshop on Interdisciplinary Applications of Graphs, Combinatorics and Algorithms' at University of Haifa, Israel \cite{ranveer2018parameterized}. The work is subsequently improved.}}

 \author{Ranveer Singh\thanks{ Technion-Israel Institute of Technology, Haifa 32000, Israel. Email: \texttt{ranveer@iitj.ac.in} (Corresponding author)}\\
  Vivek Vijay\thanks{Indian Institute of Technology Jodhpur, 342011, India.
 Email: \texttt{vivek@iitj.ac.in}}\\
RB Bapat\thanks{Indian Statistical Institute New Delhi, 110 016, India.
 Email: \texttt{rbb@isid.ac.in}}}

\begin{document}
        \maketitle
%

\begin{abstract}
Every square matrix $A=(a_{uv})\in \mathcal{C}^{n\times n}$ can be represented as a digraph having $n$ vertices. In the digraph, a block (or 2-connected component) is a maximally connected subdigraph that has no cut-vertex. The determinant and the permanent of a matrix can be calculated in terms of the determinant and the permanent of some specific induced subdigraphs of the blocks in the digraph. Interestingly, these induced subdigraphs are vertex-disjoint and they partition the digraph. Such partitions of the digraph are called the $\mathcal{B}$-partitions. In this paper, first, we develop an algorithm to find the $\mathcal{B}$-partitions. Next, we analyze the parameterized complexity of matrix determinant and permanent, where, the parameters are the sizes of blocks and the number of cut-vertices of the digraph. We give a class of combinations of cut-vertices and block sizes for which the parametrized complexities beat the state of art complexities of the determinant and the permanent.  \\
\end{abstract}
\textit{keywords:} $\mathcal{B}$-partitions, Block (2-connected component), Determinant, Permanent.
\section{Introduction}
The determinant and the permanent of a matrix are the classical problems of matrix theory \cite{abdollahi2012determinants,helton2009determinant,
bibak2013determinant,pragel2012determinants,
huang2012determinant,bibak2013determinants,
bapat2014adjacency,
hwang2003permanents,harary1969determinants,
farrell2000permanents,
wanless2005permanents,minc1984permanents}. The matrix determinant is significantly used in the various branches of science. As the zero determinant of the adjacency matrix of a graph ensures its positive nullity, the well-known problem, to characterize the graphs with positive nullity \cite{von1957spektren,bibak2013determinant}, boils down to find the determinant. Nullity of graphs is applicable in various branches of science, in particular, quantum chemistry, Huckel molecular orbital theory \cite{lee1994chemical,gutman2011nullity} and social network theory \cite{leskovec2010signed}. The number of spanning trees (forests), various resistance distances in graphs are directly related to the determinant of submatrices of its Laplacian matrix \cite{bapat2010graphs}. 
The permanent of a square matrix has significant graph theoretic interpretations. It is equivalent to find out the number of cycle-covers in the directed graph corresponding to its adjacency matrix. Also, the permanent is equal to the number of the perfect matching in the bipartite graph corresponding to its biadjacency matrix. Theory of permanents provides an effective tool in dealing with order statistics corresponding to random variables which are independent but possibly nonidentically distributed  \cite{bapat1989order}. There are various methods which use the graphical representation of matrix to calculate its determinant and permanent \cite{harary1962determinant,greenman1976graphs,
singh2017characteristic,singh2018b,bapat2010graphs}.  While determinant can be solved in polynomial time, computing permanent of a matrix is a ``$NP$-hard problem" which can not be done in polynomial time unless, $P=NP$ \cite{valiant1979complexity,wei2010matrix}.

The determinant and permanent of an $n\times n$ matrix $A$, denote by $\det(A)$, per$(A)$, respectively, are defined as
$$\det(A) = \sum_{\sigma}sgn(\sigma)a_{1\sigma(1)}\hdots a_{n\sigma(n)}
,$$ $$
\text{per}(A) = \sum_{\sigma} a_{1\sigma(1)}\hdots a_{n\sigma(n)}
,$$
where the summation is over all permutations $\sigma(1), \hdots, \sigma(n)$ of $1,\hdots,n,$ and sgn is 1 or −1 accordingly as $\sigma$ is even or odd.

Thanks to the English computer scientist Alan Turing, the matrix determinant can be computed in the polynomial time using the $LUP$ decomposition. Interestingly, the asymptotic complexity of the matrix determinant is same as that of matrix multiplication of two matrices of the same order. The theorem which relates the complexity of matrix product and the matrix determinant is as follows. 

\begin{theorem}\cite{aho1974design} \label{md}
Let $M(n)$ be the time required to multiply two $n\times n$ matrices over some ring, and $A$ is an $n\times n$ matrix. Then, we can compute the determinant of $A$ in $O(M(n))$ steps.
\end{theorem}

In general the complexity of multiplication of two matrices of order $n$ is $O(n^{\epsilon})$ where, $2\le \epsilon \le 3$. The complexities of the multiplication of two $n$ order matrices by different methods are as follows. The schoolbook matrix multiplication: $O(n^3)$,  Strassen algorithm: $O(n^{2.807})$ \cite{aho1974design}, Coppersmith-Winograd algorithm: $O(n^{2.376})$ \cite{coppersmith1990matrix}, Optimized CW-like algorithms $O(n^{2.373})$ \cite{davie2013improved,williams2011breaking,le2014powers}. The complexities of the determinant of a matrix of order $n$ by different methods are as follows. The Laplace expansion: $O(n!)$, Division-free algorithm: $O(n^4)$ \cite{rote2001division}, $LUP$ decomposition: $O(n^3)$, Bareiss algorithm: $O(n^3)$ \cite{bareiss1968sylvester}, Fast matrix multiplication: $O(n^{2.373})$ \cite{aho1974design}. Note that according to the Theorem \ref{md} the asymptotic complexity of matrix determinant is equal to that of matrix multiplication. The complexity of matrix determinant by fast matrix multiplication is same as the complexity of Optimized CW-like algorithms for matrix multiplication. However, the fastest known method to compute permanent of matrix of order $n$ is Ryser's method, having complexity $O(2^nn^2).$

In \cite{singh2017characteristic} it is shown that determinant (permanent) of a matrix can be calculated in terms of the determinant (permanent) of some subdigraphs of blocks in its digraph. The determinant (permanent) of a subdigraph means the determinant (permanent) of principle submatrix corresponding to the vertex indices of subdigraph.  First, we give some preliminaries in order to understand the procedure.

A digraph $G = (V(G), E(G))$ is a collection of a vertex set $V(G)$, and an edge set $E(G)\subseteq V(G) \times V(G)$. An edge $(u,u)$ is called a loop at the vertex $u$. A simple graph is a special case of a digraph, where $E(G)\subseteq \{(u,v): u \neq v \}$; and if $(u,v) \in E(G)$, then $(v,u) \in E(G)$. A weighted digraph is a digraph  equipped with  a weight function $f: E(G) \to \mathcal{C}$. If $V(G) = \emptyset$ then, the digraph $G$ is called a null graph. A subdigraph of $G$ is a digraph $H$, such that, $V(H) \subseteq V(G)$ and $E(H) \subseteq E(G)$. The subdigraph $H$ is an induced subdigraph of $G$ if $u, v \in V(H)$ and $(u,v) \in E(G)$ indicate $(u,v) \in E(H)$. Two subdigraphs $H_1$, and $H_2$ are called vertex-disjoint subdigraphs if $V(H_1)\cap V(H_2)=\emptyset$. A path of length $k$ between two vertices $v_1$, and $v_k$ is a sequence of distinct vertices $v_{1}, v_{2}, \dots, v_{k-1}, v_{k}$, such that, for all $i=1,2, \dots, k-1$,  either $(v_{i}, v_{i +1}) \in E(G)$ or $(v_{i + 1}, v_i) \in E(G)$. We call a digraph $G$ be connected, if there exist a path between any two distinct vertices. A component of $G$ is a maximally connected subdigraph of $G$. A cut-vertex of $G$ is a vertex whose removal results increase the number of components in $G$. Now, we define the idea of a block of a digraph, which plays the fundamental role in this article. A digraph having no cut-vertex is known as nonseparable digraph, however in this article we restrict our study to the digraphs which are not nonseparable.

    \begin{definition}{\bf Block:} 
    A block is a maximally connected subdigraph of $G$ that has no cut-vertex.
    \end{definition}
    Note that, if $G$ is a connected digraph having no cut-vertex, then $G$ itself is a block. A block is called a pendant block if it contains only one cut-vertex of $G$, or it is the only block in that component. The blocks in a digraph can be found in linear time using John and Tarjan algorithm \cite{hopcroft1971efficient}. The cut-index of a cut-vertex $v$ is the number of blocks it is associated with.

    A square matrix $A=(a_{uv})\in \mathcal{C}^{n\times n}$ can be depicted by the weighted digraph $G(A)$ with $n$ vertices. If $a_{uv} \neq 0$, then $(u,v) \in E(G(A))$, and $f(u,v) = a_{uv}$. The diagonal entry $a_{uu}$ corresponds to a loop at vertex $u$ having weight $a_{uu}$. If $v$ is a cut-vertex in $G(A)$, then we call $a_{vv}$ as the corresponding cut-entry in $A$. The following example will make this assertion transparent.

\begin{example}
    The digraphs corresponding to the matrices $M_1$, and $M_2$ are presented in Figure \ref{fig1}.
    $$M_1= \begin{bmatrix}
    0& 3& 2& 0& 0& 0& 0\\
    -7 &\color{red}5& -1& 1& -8& 0& 0\\
    2 &-1& 0& 0& 0& 0& 0\\
    0& 1& 0& 0 & 0& -3& 0\\
    0 & 12& 0& 0& 0& 1& 0\\
    0 &0& 0& 1& 1& \color{red}-4& 2\\
    0& 0& 0& 0& 0& 20& 3
    \end{bmatrix}, \ \ \ \ M_2= \begin{bmatrix}
    0& 3& 2& 0& 0& 0& 0& 0\\
    -7 &\color{red}5& -1& 1& -8& 0& 0&0\\
    2 &-1& 0& 0& 0& 0& 0&0\\
    0& 1& 0& 0 & 0& -3& 0&0\\
    0 & 12& 0& 0& 0& 1& 0&0\\
    0 &0& 0& 1& 1& \color{red}-4& 2&-2\\
    0& 0& 0& 0& 0& 20& 3&0\\
    0& 0& 0& 0& 0& -2& 0&10
    \end{bmatrix}.$$
     
    \begin{figure}
    \centering
    \begin{subfigure}{.5\textwidth}
      \centering
      \includegraphics[width=0.7\linewidth]{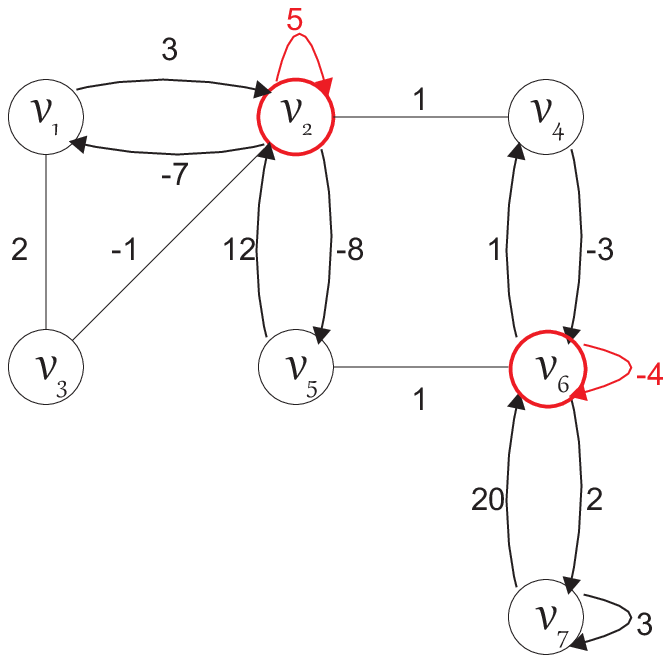}
      \caption{}
      \label{fig:sub1}
    \end{subfigure}%
    \begin{subfigure}{.5\textwidth}
      \centering
      \includegraphics[width=0.7\linewidth]{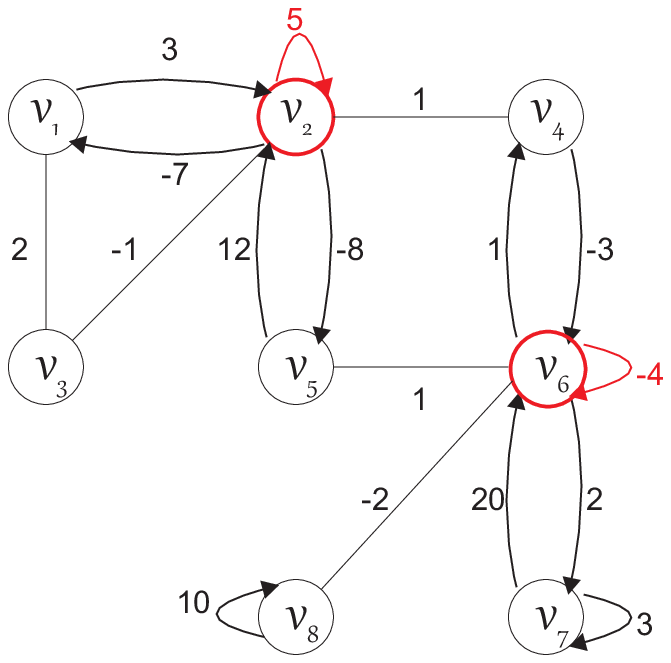}
      \caption{}
      \label{fig:sub2}
    \end{subfigure}
    \caption{(a) Digraph of matrix $M_1$ \ (b) Digraph of matrix $M_2$}
    \label{fig1}
    \end{figure}
    
    The cut-entries, and the cut-vertices are shown in red in the matrices $M_1$, and $M_2$, as well as in their corresponding digraphs $G(M_1)$, and $G(M_2)$. Note that, when $a_{uv} = a_{vu} \neq 0$, we simply denote edges $(u,v)$, and $(v,u)$ with an undirected edge $(u,v)$ with weight $a_{uv}$. As an example, in $G(M_1)$ the edge $(v_1, v_3)$ and $(v_3,v_1)$ are undirected edges.   The digraph $G(M_1)$, depicted in the Figure \ref{fig1}(a), has blocks $B_1, B_2$, and $B_3$ which are induced subdigraphs on the vertex subsets $\left\{v_1,v_2,v_3\right\}, \left\{v_2,v_4,v_5,v_6\right\}$, and $\left\{v_6,v_7\right\}$, respectively. Here, cut-vertices are $v_2$, and $v_6$ with cut-indices 2. Similarly, the digraph $G(M_2)$ in the Figure \ref{fig1}(b), has blocks $B_1, B_2, B_3$, and $B_4$ on the vertex sets $\left\{v_1,v_2,v_3\right\}, \left\{v_2,v_4,v_5,v_6\right\},\left\{v_6,v_7\right\}$, and $\left\{v_6,v_8\right\}$, respectively. Here, cut-vertices are $v_2$, and $v_6$ with cut-indices $2$, and $3$, respectively.
    \end{example}

It is to be noted that, for $G$ having $k$ blocks with number of vertices in blocks equal to $n_1,n_2,\hdots, n_k$, respectively, then the following relation holds 
\begin{equation}\label{eq1}
n=\sum_{i=1}^{k}(n_i-1)+1.
\end{equation}

Rest of the paper is organized as follows: In Section \ref{betapartition}, we first define a new partition of a digraph. An algorithm to find the $\mathcal{B}$-partitions of a digraph is given in Subsection \ref{albp}. In Subsection \ref{dp}, we give steps to find determinant and permanent of a matrix using $\mathcal{B}$-partitions from \cite{singh2017characteristic}. In Section \ref{pcmdp}, we give parametrized complexity of determinant and permanent.

\section{$\mathcal{B}$-partitions of a digraph} \label{betapartition}
    We define a new partition of digraph which is used to find determinant and permanent of a matrix.

\begin{definition} \label{def1}
     Let $G$ be a digraph having $k$ blocks $B_1, B_2, \hdots B_k$. Then, a $\mathcal{B}$-partition of $G$ is a partition in $k$ vertex disjoint induced subdigraphs $\hat{B_1}, \hat{B_2}, \hdots, \hat{B_k}$, such that, $\hat{B_i}$ is a subdigraph of $B_i$, $i=1,\hdots,k$. The $\det$-summand, and $\emph{per}$-summand of this $\mathcal{B}$-partition is  $$\prod_{i}^{k}\det (\hat{B_i}), ~\text{and}~ \prod_{i}^{k}\emph{per} (\hat{B_i}),$$ respectively, where, by convention $\det(\hat{B_i})=1, ~\text{and}~\ \emph{per}(\hat{B_i})=1 $ if $\hat{B_i}$ is a null graph. 
\end{definition}

 We give all the $\mathcal{B}$-partitions of matrices $M_1$, and $M_2$ in Figure \ref{allbm1}, \ref{allbm2}, respectively. 
     
        \begin{figure}[!htb]
             \begin{center}
     \includegraphics[width=0.5\linewidth]{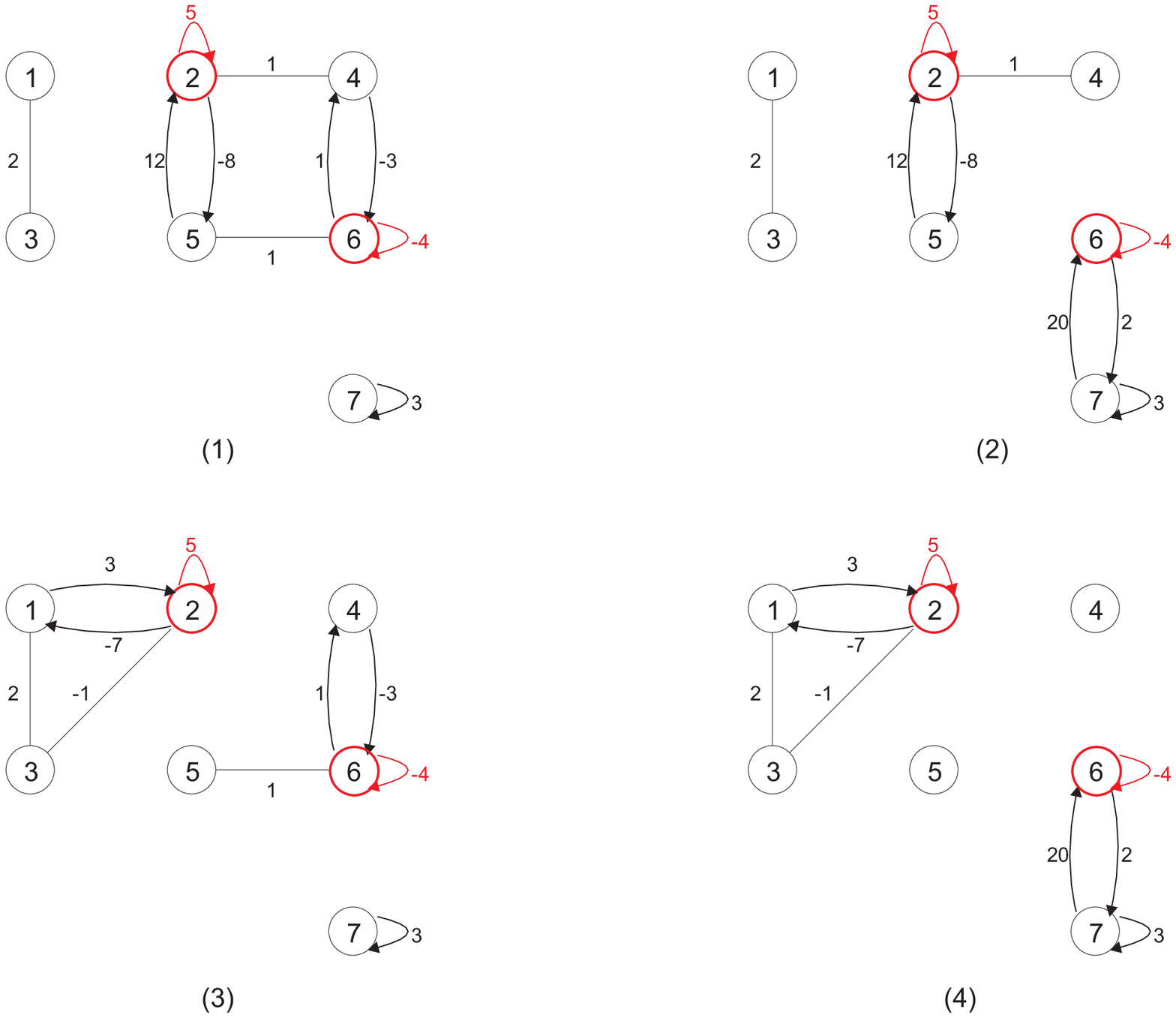}
             \end{center}
               
             \caption{$\mathcal{B}$-partitions of digraph of matrix $M_1$}
             \label{allbm1}
             \end{figure}    
         
             \begin{figure}[!htb]
                       \begin{center}
    \includegraphics[width=0.8\linewidth]{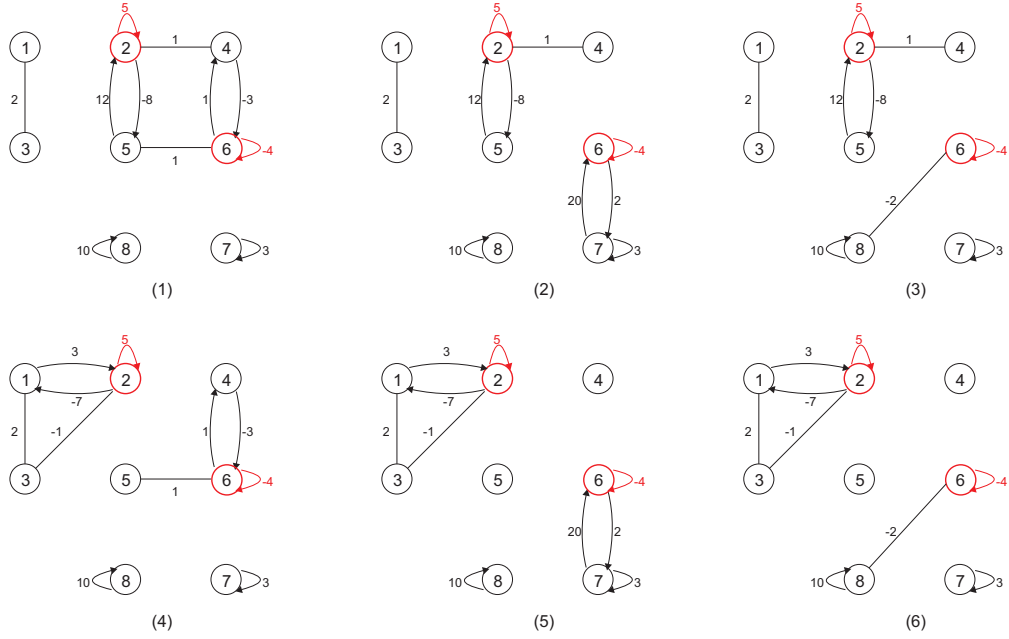}
                       \end{center}
                     \caption{$\mathcal{B}$-partitions of digraph of matrix $M_2$}
                     \label{allbm2}
         \end{figure}

\begin{corollary}
Let $G$ be a digraph having $t$ cut-vertices with cut-indices $T(1),T(2),\hdots,T(t)$, respectively. The number of $\mathcal{B}$-partitions of $G$ is  $$\prod_{i=1}^{t}T(i).$$  
\end{corollary}    
\begin{proof}
Each cut-vertex associates with an induced subdigraph of exactly one block in a $\mathcal{B}$-partition. For $i$-th cut-vertex there are $T(i)$  choices of blocks. Hence, the result follows. 
\end{proof}        
    
\subsection{Algorithm for $\mathcal{B}$-partitions}\label{albp}
If $Q$ is a subdigraph of $G$, then $G\setminus Q$ denotes the induced subdigraph of $G$ on the vertex subset $V(G)\setminus V(Q)$. Here, $V(G)\setminus V(Q)$ is the standard set-theoretic subtraction of vertex sets. Let us assume that $G$ has $k$ blocks $B_1, B_2,\hdots,B_k$. Let $t$ be the number of cut-vertices in $G$, assume them to be $v^c_1, v^c_2,\hdots,v^c_t$, where superscript $c$ denotes cut-vertex. Let $T(i)$ denote the cut-index of the cut-vertex $v^c_i$, $i=1,2,\hdots,t$. And, let $S(i)$ be the array which contains indices of the blocks to which cut-vertex $v^c_i$ belongs, and $S(i,j)$ denote its $j$-th element.

\begin{example} \label{info}
For the digraph of matrix $M_1$, $t=2$. Let $v^c_1=v_2, v^c_2=v_6 $ then, $T(1)=2, T(2)=2, S(1)=[1,2]$ and $S(2)=[2,3]$. Similarly, for the digraph of matrix $M_2$, $t=2$. Let  $v^c_1=v_2, v^c_2=v_6 $ then, $T(1)=2, T(2)=3, S(1)=[1,2],$ and $S(2)=[2,3,4]$. 
\end{example}

The algorithm to find all the $\mathcal{B}$-partitions of a digraph $G$ is given in Algorithm \ref{balgo}. For a $\mathcal{B}$-partition the algorithm associates each cut-vertex to exactly one of its associated block and remove from rest of the blocks, thus it recursively finds all the $\mathcal{B}$-partitions. Interested readers can see the Matlab codes for the determinant\footnote{\url{<https://in.mathworks.com/matlabcentral/fileexchange/62442-matrix-det--a-->}} and the permanent\footnote{\url{<https://in.mathworks.com/matlabcentral/fileexchange/62443-matrix-per--a-->}} using the Algorithm\footnote{Although these version of codes are not so fast, its due to the way codes are written, not due to the algorithm.} \ref{balgo}.

\begin{algorithm} 
\SetAlgoLined
\KwResult{$\mathcal{B}$-partition of $G$\\ Parameters and initialisation: $t$ is number of cut-vertices, let they be $v^c_1, v^c_2,\hdots,v^c_t$.  For $i=1,2,\hdots,t$, assume that cut-vertex $v^c_i$ has cut-index equal to $T(i)$.  $S(i)$ is the array which contains indices of the blocks to which cut-vertex $v^c_i$ belongs, and $S(i,j)$ denotes its $j$-th element.  $X$ is initialised to zero column vector of order $t\times 1$, and $X(i)$ denotes its $i$-th entry.  Function $\mathcal{B}$-part$(it,X,t,S(it),T(it))$ gives all the $\mathcal{B}$-partitions recursively.}
 $it=1$ \;
 $\mathcal{B}$-part$(it,X,t,S(it),T(it))$
 
  \eIf{$(it>t)$}{   $S(it)=S(it-1)$ \; $T(it)=T(it-1)$ \;
  \For{$i=1:t$}{\For{$j=1:T(i)$}{Remove cut-vertex $v^c_i$ from blocks $B_{S(i,j)}$ when  $X(i)\neq S(i,j) $.}}
   The resulting subgraphs of blocks will give a $\mathcal{B}$-partition. save it\;
   return }
   {\For{$i=1:T(it)$}{$X(it)=S(it,i)$\;  $\mathcal{B}$-part$(it+1,X,t,S(it+1),T(it+1))$}
  } 
 \caption{Algorithm for $\mathcal{B}$-partitions.}\label{balgo}
\end{algorithm}

\emph{Explanation and the complexity:} The function $\mathcal{B}$-part$(it,X,t,S(it),T(it))$ recursively assigns the cut-vertices $v^c_1, v^c_2,\hdots,v^c_t$ to one of their blocks, respectively, and removes them from rest of their blocks. The vector $X$ does this job. Whenever $it>t$ the algorithm assign the cut-vertex $v^c_i$ to the block with index $X(i)$ and removes it from all of its remaining blocks. Which gives a $\mathcal{B}$-partition. The Algorithm then execute from where it previously left, that is, to line 14.  For each cut-vertex $v^c_i, i=1,\hdots,t$, algorithm executes for $T(i)$ times, hence the total steps taken are $\prod_{i=1}^tT(i)$.    
As an example, Figure \ref{fig2} tells how the vector $X$ is updating for the matrix $M_1$ to calculate its $\mathcal{B}$-partitions using the information provided in Example \ref{info}. We have $v^c_1=v_2, v^c_2=v_6$, and $v^c_1$ is in blocks $B_1, B_2$ while $v^c_2$ is in blocks $v^c_2$. The values of $X$ in stages (3)-(6) gives all the four $\mathcal{B}$-partitions. For example in stage (3), $v_1^c$ is assigned to $B_1$ and removed from $B_2$, while, $v_2^c$ is assigned to $B_2$ and removed from $B_3$. It gives the $\mathcal{B}$-partition $B_1, B_2\setminus v_1^c, B_3\setminus v_2^c$, see Figure \ref{allbm1}(3).

\begin{figure} 
\hspace{20mm}\begin{tabular}{ c c c c}
  &  &  & $X$\\ 
 $v_1^c$ & $B_1$ & $B_2$ & 0\\  
 $v_2^c$ & $B_2$ & $B_3$ & 0 \\
  & & (1) &   
\end{tabular}, \ \ \ \ \ \begin{tabular}{ c c c c}
  &  &  & $X$\\ 
 $v_1^c$ & \color{red}{$B_1$} & $B_2$ & \color{red}{1} \\  
 $v_2^c$ & $B_2$ & $B_3$ & 0 \\  & & (2) &   
\end{tabular}, \ \ \ \ \ \begin{tabular}{ c c c c}
  &  &  & $X$\\ 
 $v_1^c$ & \color{red}{$B_1$} & $B_2$ & \color{red}{1} \\  
 $v_2^c$ & \color{red}{$B_2$} & $B_3$ & \color{red}{2} \\
  & & (3) &  
\end{tabular}, 

\hspace{20mm}\begin{tabular}{ c c c c}
  &  &  & $X$\\ 
 $v_1^c$ & \color{red}{$B_1$} & $B_2$ & \color{red}{1} \\  
 $v_2^c$ & $B_2$ & \color{red}{$B_3$} & \color{red}{3} \\
  & & (4) &  
\end{tabular}, \ \ \ \ \  \begin{tabular}{ c c c c}
  &  &  & $X$\\ 
 $v_1^c$ & $B_1$ & \color{red}{$B_2$} & \color{red}{2} \\  
 $v_2^c$ & \color{red}{$B_2$} & $B_3$ & \color{red}{2} \\
  & & (5) &  
\end{tabular}, \ \ \ \ \  \begin{tabular}{ c c c c}
  &  &  & $X$\\ 
 $v_1^c$ & $B_1$ & \color{red}{$B_2$} & \color{red}{2} \\  
 $v_2^c$ & $B_2$ & \color{red}{$B_3$} & \color{red}{3} \\
  & & (6) &  
\end{tabular}
\caption{ Updation of vector $X$ in Algorithm \ref{balgo} for the matrix $M_1$. The index of selected block (shown in red) for a cut-vertex is copied to the correspondind entry in $X$.} \label{fig2}
\end{figure}

 For matrix $M_1$, Algorithm \ref{balgo} begins the execution by calling the function $\mathcal{B}$-part$(1, X,2, S(1),T(1))$ in line 2. As $it=1\ngtr t$ execution will proceed from line 14. Line 15 will make $X(1)=S(1,1)=1$, that is $X=[1, 0]^T$.  At line 16 the function $\mathcal{B}$-part $(2,X,2, S(2),T(2))$ will be called, which again starts from line 14, as still $it=2\ngtr t$. Line 15 makes $X(2)=S(2,1)=2$, that is $X=[1,2]^T$.  Line 16 then call the function  $\mathcal{B}$-part $(3,X,2, S(3),T(3))$. As $3>t$, now the execution starts from line 4. $S(it), T(it)$ are restored to just previous values. By line 8, we need to remove vertex $v^c_1$, that is vertex $v_2$ from block $B_2$. Similarly, we need to remove vertex $v^c_2$, that is vertex $v_6$ from block $B_3$. The resulting subgraphs of blocks are $B_1, B_2\setminus v_2, B_3\setminus v_6$, which gives a $\mathcal{B}$-partition, see Figure \ref{allbm1}(3). In line 12, execution will return to line 14, where now $i=2, it=2$, thus $X(2)=S(2,2)=3$, that is $X=[1,3]^T$. Line 16 then call the function  $\mathcal{B}$-part $(3,X,2, S(3),T(3))$. This time we need to remove vertex $v^c_1$, from block $B_2$ and vertex $v^c_2$ from block $B_2$. The resulting subgraphs of blocks are $B_1, B_2\setminus (v_2,v_6), B_3$, which gives an another $\mathcal{B}$-partition, see Figure \ref{allbm1}(4). In line 12 execution will return to execution where $it=1$, and $i=2$. Proceeding as before we get two more $\mathcal{B}$-partitions $B_1\setminus v_2, B_2, B_3\setminus v_6$ and $B_1\setminus v_2, B_2\setminus v_6, B_3$, see Figure \ref{allbm1}((1),(2), respectively).

In the following subsection, we give a procedure to find the determinant and permanent of a given square matrix using the $B$-partitions of its digraph \cite{singh2017characteristic}.

\subsection{Determinant and permanent of matrix} \label{dp}
Let $G$ be a digraph having $k$ blocks $B_1, B_2, \hdots, B_k$. Let $t_{nz}$ cut-vertices in $G$  have nonzero weight on their loops. Let the weights of loops at these vertices be $\alpha_1, \alpha_2, \hdots, \alpha_{t_{nz}}$, respectively. Also, assume that these cut-vertices have cut-indices $T(1), T(2), \hdots, T(t_{nz})$, respectively. Then, the following are the steps to calculate the determinant (permanent) of $G$.
    
\begin{enumerate}
 \item For $q=0,1,2,\dots,t_{nz},$\begin{enumerate}
\item Select any $q$ cut-vertices at a time which have nonzero loop weights. In each $\mathcal{B}$-partition, remove these $q$ cut-vertices from subdigraph to which they belong. 
\begin{enumerate}
\item For all the resulting partitions, sum their $\det$-summands ($\text{per}$-summands).  
Multiply the sum by $\prod_{i}\frac{(-\alpha_i)(T(i)-1)}{T(i)}$ where, $\alpha_i$, and $T(i)$ are the weight and the cut-index of the removed $i$-th cut-vertex, respectively, for $i=1,2,\dots,q$.  

\item For all possible $\binom{t_{nz}}{q}$ combinations of removed cut-vertices, sum all the terms in i.
 
\end{enumerate}  
\end{enumerate}
     
    \item Sum all the terms in 1.
    
\end{enumerate}
    
\begin{example} \label{example}
The contribution of different summands to calculate the determinant of $M_1$ is as follows. 
    \begin{enumerate}
    \item For $q=0,$
 that is, without removing any cut-vertex we get the following part of $\det(M_1)$.
    \begin{equation*}
    \begin{split}
\det(B_1)\det(B_2\setminus v_2)\det(B_3\setminus v_6)+\det(B_1)\det(B_2\setminus (v_2,v_6))\det(B_3)\\+\det(B_1\setminus v_2)\det(B_2)\det(B_3\setminus v_6)+\det(B_1\setminus v_2)\det(B_2\setminus v_6)\det(B_3) 
    \end{split}
    \end{equation*}
    
    \item For $q=1.$ 
  \begin{enumerate}
 \item On removing $v_2:$ recall that the loop-weight of $v_2$ is $5$ and its cut-index is 2. Removing the cut-vertex $v_2$ we get the following part, 
    \begin{equation}
 -5\Big(\det(B_1\setminus v_2)\det(B_2\setminus v_2)\det(B_3\setminus v_6)+\det(B_1\setminus v_2)\det(B_2\setminus (v_2,v_6))\det(B_3)\Big).   
 \end{equation}
    
    \item On removing $v_6:$ the loop-weight of $v_6$ is $-4$ and its cut-index is 2. Removing the cut-vertex $v_6$ we get the following part, \begin{equation}
    4\Big(\det(B_1)\det(B_2\setminus (v_2,v_6))\det(B_3\setminus v_6)+\det(B_1 \setminus v_2)\det(B_2\setminus v_6)\det(B_3\setminus v_6)\Big).
    \end{equation}
    
  \end{enumerate}

    \item For $q=2,$ that is, removing the cut-vertices $v_2$, and $v_6$ we get the following part,\begin{equation}
    -5\times 4\det(B_1\setminus v_2)\det(B_2\setminus (v_2,v_6))\det(B_3\setminus v_6).
    \end{equation}
    
    Adding (1), 2(a), 2(b) and (3) give $\det(M_1).$ 
    \end{enumerate}

\end{example}

The permanent of the matrix $M_1$ is similarly calculated using the per-summands instead of $\det$-summands. 

\section{Parametrized complexity of determinant and permanent}\label{pcmdp}
From the last section, we saw how the subdigraphs of blocks are used to calculate the determinant and the permanent of a given square matrix. Before we proceed to find parametrized complexity of matrix determinant, let us observe the following. Let $A$ be the following square matrix of order $r$, as follows 
    $$A=\begin{bmatrix}
    A_1 & \mathbf{b}\\ \mathbf{c}& d
    \end{bmatrix}.$$
where, $A_1$ is the principle submatrix of order $r-1$, $\mathbf{b}$ is the column vector of order $r-1$, $\mathbf{c}$ is the row vector of order $r-1$, and $A(r,r)=d$. 
\begin{enumerate}

\item If $A_1$ is invertible then by Schur' complement for determinant \cite{bapat2014adjacency},
$$\det (A)=\det (A_1) \det(d-\mathbf{c}A^{-1}_1\mathbf{b}).$$
    
\item If $A_1$ is not invertible, and $d$ is nonzero then, $$\det (A)=d \times \det(A_1-\mathbf{b}\frac{1}{d}\mathbf{c}).$$
\item If $A_1$ is not invertible, and $d$ is zero then, $$\det (A)= \det \begin{bmatrix}
A_1 & \mathbf{b}\\ \mathbf{c}& 0
\end{bmatrix}=\det \begin{bmatrix}
A_1 & \mathbf{b}\\ \mathbf{c}& 1 \end{bmatrix}-1\times \det( A_1)=\det(A_1-\mathbf{b}\mathbf{c}).$$ 
\end{enumerate}
Fortunately, complexity of inverse of a square matrix of order $r$ is also $O(r^{\epsilon})$ \cite{aho1974design}. 
Thus, in all the above cases, we observe that the $\det(A)$ can be calculated in terms of the determinant of lower order matrices. For example in Subsection \ref{dp}, in Example \ref{example}, $\det[4,5,6]$ can be calculated using $\det[4,5]$. Now, in $G$ consider a block $B_i$ has $t_i$ number of cut-vertices. We can first calculate the determinant of resulting subgraph after the removal of all the $t_i$ cut-vertices from $B_i$. Then this determinant can be used to calculate the determinant of resulting subgraph after removal of any $t_i-1$ cut-vertices. In this way, in general determinant of resulting subgraph after removal of $i$ cut-vertices can be used to calculate the determinant of resulting subgraph after removal of $i-1$ cut-vertices. Now we proceed towards the parameterized complexity.

Let a digraph $G$ have $k$ blocks $B_1, B_2,\hdots, B_k$. Let the sizes of the blocks be $n_1,n_2,\hdots, n_k$, and the number of cut-vertices in the blocks be $t_1, t_2,\hdots, t_k$, respectively. Then, during the calculation of determinant of $G$, for a particular block $B_i$, the determinant of subdigraph of size $n_i-j$ is being calculated $\binom{t_i}{j}$ times. Also, in the view of the above observation, in order to calculate the determinant of subdigraph of order $n_i-j$ we can use determinant of subdigraph of order $n_i-j-1$. Hence, the complexity of calculating the determinant is  

$$O\Big(\sum_{i=1}^{k}\sum_{j=0}^{t_i}\binom{t_i}{j}(n_i-j-1)^{\epsilon}\Big).$$

As $(n_i-j-1)< n_i$, an upper bound of the above complexity is 
\begin{equation}\label{cd}
O\Big(\sum_{i=1}^{k}2^{t_i}
n_i^{\epsilon}\Big).
\end{equation}

The obvious pertinent question that follows is, for which combinations of $n_i, t_i$ the complexity beats the state of art complexity $O(n^{\epsilon})$. Thus, we need to solve the following inequality

\begin{equation}
\sum_{i=1}^{k}2^{t_i}n_i^{\epsilon}< n^{\epsilon}.
\end{equation}

Similarly, for the permanent the parametrized complexity is 

$$ O\Big(\sum_{i=1}^{k}\sum_{j=0}^{t_i}\binom{t_i}{j}2^{(n_i-j)}(n_i-j)^2\Big).$$
A upper bound of the above complexity is 
\begin{equation}\label{cp}
O\Big(\sum_{i=1}^{k}2^{t_i}2^{n_i}n_i^2\Big).
\end{equation} 
Similar to the determinant, the question that follows for the permanent is, for which combinations of $n_i, t_i$ the above complexity beats $O(2^{n}n^2)$. Thus is, we need to solve the following inequality
\begin{equation}
\sum_{i=1}^{k}2^{t_i}2^{n_i}n^2_i< 2^nn^2.
\end{equation}

\begin{figure}[hbt]
\centering
  \includegraphics[width=12cm]{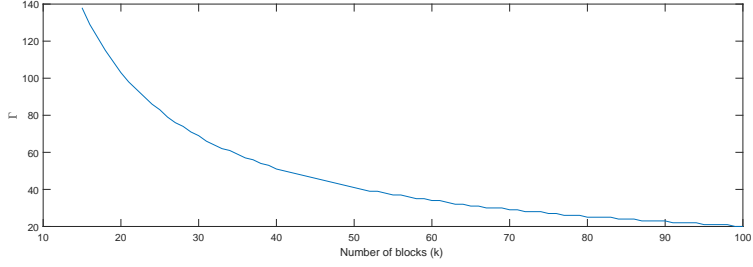}
  \caption{Example. $n=1000, \Delta=200, \epsilon=2.373$. The points below the curves gives the combinations of $\Gamma, k$ for which the proposed algorithm for determinant will work faster than the state of art algorithms.}\label{dtfig}
\end{figure}

\subsection{Parametrized complexities}
Let $\Gamma$ be the largest of the numbers of the cut-vertices in any block, that is, $\Gamma=max\{t_1,t_2,\hdots t_k\}.$ Let $\Delta$ be the size of the largest block, that is, $\Delta=max\{n_1,n_2,\hdots, n_k\}.$
From Expression (\ref{cd})
\begin{equation}
\sum_{i=1}^{k} 2^{t_i}n_i^{\epsilon}\leq 2^{\Gamma}\sum_{i=1}^{k}n_i^{\epsilon}\leq k2^{\Gamma}\Delta^{\epsilon}
\end{equation}
In order to beat $n^{\epsilon},$ 
\begin{equation}
k2^{\Gamma}\Delta^{\epsilon}\leq n^{\epsilon}
\end{equation}
Taking logarithm on both the sides, we have 
\begin{equation}
\Gamma \ln 2 \leq \ln\bigg(\frac{1}{k}\Big(\frac{n}{\Delta}\Big)^{\epsilon}\bigg), 
\end{equation} 
that is 
\begin{equation}
\Gamma=O\Bigg(\ln\bigg(\frac{1}{k}\Big(\frac{n}{\Delta}\Big)^{\epsilon}\bigg)\Bigg).
\end{equation}
For example in Figure (\ref{dtfig}), for $n=1000, \Delta=200$, $\epsilon=2.373$, the combinations of $\Gamma, k$ are given for which the proposed algorithm for determinant will work faster than the state of art algorithms.

Similarly for the permanent from the Expression (\ref{cp})

\begin{equation}
\sum_{i=1}^{k}2^{t_i}2^{n_i}n_i^2\leq k2^{\Gamma}2^{\Delta}\Delta^2 
\end{equation}

in order beat $2^nn^2$ we have 

\begin{equation}
k2^{\Gamma}2^{\Delta}\Delta^2 \leq 2^nn^2.
\end{equation}

Taking logarithm on both the sides 
\begin{equation}
\Gamma \ln2\leq \ln\Big(\frac{2^nn^2}{k2^\Delta\Delta^2}\Big),
\end{equation}
that is 
\begin{equation}
\Gamma=O\Bigg( \ln\Big(\frac{2^nn^2}{k2^\Delta\Delta^2}\Big)\Bigg).
\end{equation}

It is clear from the above equation that for the permanents, if $\Delta<n$, then $2^\Delta \ll 2^n,$ the proposed algorithm will work faster than the state of art for almost all the cases.  

\textbf{Acknowledgment.} The first author acknowledges support from the JC Bose Fellowship, Department of Science and Technology, Government of India. The authors are thankful to Annuay Jayaprakash, Vijay Paliwal for the help in the algorithm. 
\bibliographystyle{plain}
    \bibliography{MDPC}

 \end{document}